\documentclass[a4paper,11pt]{article}

\usepackage[latin1]{inputenc}
\usepackage{color}
\usepackage{amsmath}
\usepackage{fullpage}
\usepackage{graphicx}
\usepackage{float}
\usepackage{pst-all}
\usepackage{moreverb}
\usepackage{verbatim}
\usepackage{pstricks}
\newfloat{algorithm}{thp}{lop}[section]
\floatname{algorithm}{Algorithm}
\newfloat{module}{thp}{lop}[section]
\floatname{module}{Module}

\newcommand{\remove}[1]{}
\newtheorem{definition}{Definition}

\newtheorem{lemma}{Lemma}

\newcommand{\qedsymb}{\hfill{\rule{2mm}{2mm}}}
\newenvironment{proof}
{\begin{trivlist}
\item[\hspace{\labelsep}{\bf\noindent Proof: }]
}
{\qedsymb\end{trivlist}}

\newlength {\squarewidth}

\newcounter{linecounter}

\begin{document}

\author{Davide Canepa\footnote{LIP6, Univ. Pierre \& Marie Curie - Paris 6, LIP6-CNRS UMR 7606, France.} \hspace{1cm} Xavier Defago\footnote{JAIST, Japan Advanced Institute in Science and Telecomunication, Japan} \hspace{1cm} Taisuke Izumi\footnote{Graduate School of Engineering, Nagoya Institute of Technology, Japan} \hspace{1cm} Maria Potop-Butucaru$^*$}
 
\date{ }
\title{Emergent Velocity Agreement in Robot Networks}
\maketitle

\thispagestyle{empty}
\begin{abstract} 
In this paper we propose and prove correct 
a new self-stabilizing velocity agreement (flocking) algorithm for oblivious and 
asynchronous robot networks. 
Our algorithm allows a flock 
of  uniform robots to follow a 
flock head emergent during the computation whatever its direction in plane. 
Robots are asynchronous, oblivious and do not share a common coordinate system.
Our solution includes three modules architectured as follows: 
creation of a common coordinate system that also allows the emergence of a flock-head, 
setting up the flock pattern and moving the flock. 
The novelty of our approach steams in identifying the necessary 
conditions on the flock pattern placement and the velocity of the flock-head (rotation, translation or speed) that allow the flock 
to both follow the exact same head and to preserve the flock pattern. Additionally, our system is {\it self-healing} and {\it self-stabilizing}. In the event of the head leave (the leading robot disappears or is damaged and cannot be recognized by the other robots) the flock agrees on another head and 
follows the trajectory of the new head. Also, robots are {\it oblivious} (they do not recall the result of their previous computations) and we make no assumption on their initial position. The step complexity of our solution is $O(n)$.
\end{abstract}

\newpage
\setcounter{page}{1}
\section{Introduction}
Flocking gained recently increased attention in diverse areas such as biology, economy, language study or agent/sensor networks. In biology, flocking refers the coordinate behaviour of a group of birds or animals when they sense some imminent threat or lookup for food. In economy, the emergent behaviour that regulates the stock markets can be seen as a form of flocking. The emergency of a common language in primitive societies is also an instantiation of flocking. 

In the context of robot networks, flocking is the ability of a group of robots to coordinate and move in the plane or space. This coordinated motion has several civil and military applications ranging from zone exploration to space-crafts self-organization.

In distributed robot networks there are two types of agreement problems that have been studied so far: point or pattern agreement  and velocity agreement.  
{\it Point agreement} (gathering or convergence) aims at instructing robots to reach an exact or approximate common point not known {\it a priori}. The dual is the scattering problem where robots are instructed to reach different positions in the plane. Furthermore, 
pattern agreement deals with instructing robots to eventually arrange in a predefined shape (i.e. circular, rectangular etc). 

{\it Velocity agreement} or flocking refers the ability of robots to coordinate and move in 2D or 3D spaces without any external intervention. 
The literature agrees on two different strategies to implement flocking.
The first strategy is based on a predefined hierarchy. That is, there is an a priori leader clearly identified in the group that will lead the group and each group member will follow the leader trajectory. An alternative is to obtain an emergent coordination of the group without a predefined leader.  The difficulty of this approach comes from the permanent stress for connectivity maintenance. That is, if the flock splits then it may never converge to a single flock. 

In this paper we are interested in uniform flocking strategies. That is, the head of the flock emerges during the computation hence the solution is self-healing. Also, the flock does not know a priori the motion trajectory. That is, the head of the flock will lead the flock following its own trajectory that may be predefined or decided on the fly.  
Our work is developed in the asynchronous CORDA model $\cite{flocchini00distributed,prencipe01corda}$ 
one of the two theoretical models proposed so far for oblivious distributed robot networks.

The first distributed model for robot networks, SYm,  was introduced by Suzuky and Yamashita  
$\cite{suzukiYama94,suzuki96distributed,suzuki99distributed}$. In SYm model 
robots are oblivious and perform a cycle of elementary actions as follows : 
observation (the robot observes the environment), computation 
(the robot computes its next position based on the information collected in 
the observation phase) and 
motion (the robot changes its position to the newly computed position).  
In this model robots cannot be interrupted during the execution of a cycle. 
The CORDA model breaks the execution cycle in elementary actions. That is, a robot can be 
activated/turned off while it executes a cycle. Hence, robots are not any more synchronized.

The \textit{flocking problem} although 
largely discussed for real robots (\cite{8971839,FlckOstAvoid,8348518}) 
was studied from distributed theoretical point of view mainly by 
Prencipe $\cite{gervasi03coordination,gervasi01flocking}$. The authors propose 
non-uniform algorithms where robots play one of the following roles: leader 
or follower. The leader is unique and all the followers know the leader robot.
Obviously, when the leader crashes, disappears or duplicates 
the flock cannot finish its task. In \cite{CG07} we extend the results in \cite{prencipe00achievable} and propose a probabilistic flocking architecture. 
However, we make the assumption that the leader and consequently the flock do not change their direction and trajectory. Our current approach is different, the leader is not known a priori 
but it will emerge during the computation. 
When the current leader disappears or is damaged and not recognized as a correct robot, the 
other robots in the system agree on another leader and the flock can finish its task.
Furthermore, the flock can change both its direction and trajectory in order to agree with the emergent leader velocity.

Fault-tolerant (but not self-stabilizing) flocking has been addressed in \cite{SIW09,YSDT11}.
 In \cite{SIW09} the authors propose a fault tolerant flocking algorithm in the SYm model using a leader oracle and a failure detector. In our solution the leader or the head of the flock emerges during the computation. Also, our solution works in asynchronous settings and their decision is solely based on their current observation.
In \cite{YSDT11} the authors also propose a fault tolerant flocking.
It is assumed the SYm model (awaken robots execute their operations in synchronous steps) and the  k-bounded scheduler (in between two actions of a robot any other robot executes at most $k$ actions). Also the solution needs 
agreement on one axis, agreement on chirality and non-oblivious robots. Contrary to this approach, our solution does not need any a priori agreement on axis or chirality. Moreover, we assume oblivious robots.

Several works from robotics propose recently heuristics for flocking(e.g. \cite{LC08a,MSC10}).
In \cite{LC08a}, for example, the authors propose a solution for non-uniform flocking.  In their proposal  the leader has to execute a different strategy than the rest of the flock. Hence, the system is not uniform.


\paragraph{Our contribution.}
In this paper we propose and prove correct a new asynchronous flocking algorithm 
in systems with oblivious and uniform robots.  
Additionally, we identify the necessary conditions on the flock pattern placement and the flock head velocity (rotation, translation, speed) to allow the flock to maintain the same leader and the common coordinate system and also to preserve the motion pattern.
Our solution is composed of three asynchronous {\it self-stabilizing} and {\it self-healing} phases. First robots agree on 
a common coordinate system and a leader. Once this phase is finished 
robots form a flocking pattern and move preserving the 
same system of coordinates and the same leader. In the event of the head leave (the leading robot disappears or is damaged and cannot be recognized by the other robots) the flock agrees on another head and 
follows the trajectory of the new head. Also, robots are {\it oblivious} (they do not recall the result of their previous computations) and we make no assumption on their initial position. The complexity of our solution is $O(n)$.

\paragraph{Paper organization.}
The paper is organized as follows: Section \ref{sec:model} defines the model of the system, Section \ref{sec:flocking-problem} specifies the problem based on the flocking informal definitions in different areas ranging from biological systems to space navigation. In this section we also propose a brief description of our system architecture. Section 
\ref{sec:common-coordinates} sets up the common coordinate system, 
Section \ref{sec:flocking} details the formation of the flocking pattern and the necessary conditions on the flock placement, 
Section \ref{sec:motion} proposes the rules for moving the flock and identifies the necessary conditions on the flock head velocity.

\section{Model}
\label{sec:model}
Most of the notions presented in this section are borrowed 
from \cite{gervasi01flocking,prencipe01corda}. 
We consider a system of autonomous mobile robots that work in the CORDA model \cite{prencipe01corda}.
 
Each robot is capable of observing its surrounding, computing a destination
based on what it observed, and moving towards the computed destination:
hence it performs an (endless) cycle of observing, computing, and moving.
Each robot has its own local view of the world. 
This view includes a local Cartesian coordinate system having an origin, a unit of
length, and the directions of two coordinate axis (which we will refer to as
the x and y axis), together with their orientations, identified as the positive
and negative sides of the axis. 

The robots are model as units with computational capabilities, which are
able to freely move in the plane. They are equipped with sensors that let
each robot observe the positions of the others with respect to their local
coordinate system. Each robot is viewed as a point, and can see all the other
robots in the flock.

The robots act totally independent and asynchronously from each other,
and do not rely on any centralized directives, nor on any common notion of
time. Furthermore, they are oblivious, meaning that they do not remember
any previous observation nor computations performed in the previous steps.
Note that this feature combined with no assumptions on the initial position of the robots 
gives to the algorithms 
designed in this model the nice property of {\it self-stabilization} \cite{dol00}. That is, every decision taken by a robot does not depend on what happened previously
in the system and robots do not use potentially corrupted data stored
in their local memory. 

Robots in the flock are anonymous (i.e. they are a priori indistinguishable by their appearances and they do not have any kind of identifiers
that can be used during the computation).  Moreover, there are no explicit direct
means of communication; hence the only way they have to acquire information
from their fellows robots is by observing their positions.
They execute the same algorithm (the system is uniform), which takes as input the observed positions
of the fellow robots, and returns a destination point towards which they target their move.

Summarizing, each robot moves totally independent and asynchronously
from the others, not having any bound on the time it needs to perform a move,
hence a cycle; therefore, a robot can be seen while it moves. In addition, 
robots are oblivious, and anonymous.  We make no assumption on the initial position of robots or 
a common coordinate system.

\section{The flocking problem} 
\label{sec:flocking-problem}
Reynolds proposed in the mid of 80's  three rules that have to be respected by any algorithm that simulates a flock-like behaviour. He successfully applied these rules in designing several animations.  At that time the flock entities were 
called boids and the model was as follows: each boid has the ability to sense its local neighbours; each boid can sense the whole environment, all boids recalculate their current state simultaneously once each time unit during the simulation.

In this model, according to Reynolds, the flocking rules are as follows: 
  
\begin{itemize}
\item Separation: steer to avoid crowding local flock-mates.
\item Alignment: steer towards the average heading of local flock-mates.
\item Cohesion: steer to move toward the average position of local flock mates.
\end{itemize}

 Interestingly, the model proposed by Reynolds is similar to the previously described SYm model. Robots can sense the environment (the other robots in the system) and they periodically and simultaneously recalculate their state.  However, in CORDA model (the one used in the current work) the computation is asynchronous. 
Nevertheless, the main and important difference with respect to the Reynolds assumptions is the impossibility to use the  history of the computation in order to implement the flocking rules. Note that the second rule of Reynolds indirectly use this information. Therefore,  in the case of robot networks these rules should be adapted. 

In distributed robot networks acceptance, flocking allows a group or a formation of robots to change their position either by  following a pre-designated or an emergent leader. In this case the flocking is reffered as {\it uniform}. Intuitively, a flock is a group of robots that move in the plane in order to execute a task while maintaining a specific formation. This informal definition  implicitly assumes the existence of an 
unique leader of the flock that will lead the group during the task execution and the existence of a flocking pattern. Also it is assumed a virtual link between the head and the rest of the group. Therefore, three elements seem to be essential in the definition of the flocking : the head or the leader of the group, the pattern and the orientation of the pattern with respect to the leader. Based on these elements, flocking can be seen as the motion of the virtual rigid formed by the flock and its head following a trajectory predefined or defined on-the-fly. It follows that both the flock and its head periodically synchronize their velocity 
in order to maintain the flock.
In the following we specify the uniform flocking problem (i.e. the leader emerges during the computation). We first recall the definition of leader election and pattern formation. According to a recent study, \cite{DPV10,FPSW08}, pattern formation and leader election are related problems. Our specification naturally extends this observation to the flocking problem.

\begin{definition}[Leader Election]\cite{FPSW08}
Given the positions of $n$ robots in the plane, the $n$
robots are able to deterministically agree on the same robot called the leader. 
\end{definition}

\begin{definition}[Pattern Formation]\cite{FPSW08} 
The robots have in input the same pattern, called the
target pattern $\mathcal F$, described as a set of positions in the plane given in 
lexicographic order (each robot sees the same pattern according to the direction and orientation
of its local coordinate system). They are required to form the pattern: at the end 
of the computation, the positions of the robots coincide, in everybody's local view, 
with the positions of $\mathcal F$, where $\mathcal F$ may be translated, rotated, and scaled in each local coordinate system. 
\end{definition}

\begin{definition}[Uniform Flocking]
Let $r_1 \ldots r_n$ be a set of robots and let $\mathcal{F}$ be the flocking pattern. The set of robots satisfy the flocking specification if the following properties hold:
\begin{itemize}
\item {\bf head/leader emergence} eventually robots agree on an unique head (leader), $r_1$;
\item {\bf pattern emergence} eventually robots, $r_2, \ldots, r_n$, form the pattern $\mathcal{F}$;
\item {\bf velocity agreement} after any modification of the $r_0$ position, robots in the pattern rotate and translate $\mathcal{F}$ in order to converge to the same relative position and orientation between $r_0$  and $\mathcal{F}$ as it was before the modification.
\item {\bf no collision} any robot motion is collision free.
\end{itemize}
\end{definition}

Note the common flavour between the Reynolds rules and the above properties. {\it No collision} property corresponds to the separation rule. {\it Velocity agreement} corresponds to the alignment rule and finally {\it leader and pattern emergence} are similar to the cohesion property. 
  
 In the following we combine three different tasks to solve the uniform flocking in systems where robots are asynchronous, do not share the same coordinate systems, are oblivious and uniform.
First, we design a novel strategy for equipping a set of robots with a common coordinate systems.
To this end we propose a probabilistic strategy that creates two singularity points.
Then, we combine this module with existing probabilistic election strategies (\cite{CG07,DP07}) in order to create the third singularity point. The motion of this third point will eventually designate the head of the flock (robot $r_0$) and the orientation of the common coordinate system.
Then, the emergent common coordinate system is further used by all the robots but $r_0$ to arrange themselves in a flocking pattern, $\mathcal{F}$, that  will further follow the head $r_0$. During the pattern motion, both the head and the common coordinate system are preserved.

\section{Common Coordinate System and Flock Head Emergence}
\label{sec:common-coordinates}
The construction of a common coordinate system is as follows. First, robots 
agree on one axis, then they agree on the second axis,  orientation of axis and the head of the flock. Due to space limitation, the details  and the correctness proof of these algorithms are proposed in the Annexes (section \ref{sec:common-coordinates-details}).

\paragraph{Agreement on the first axis.} Note that
one axis is defined by two distinct points. 
The algorithm idea is very simple: robots 
compute the barycentre of their convex hull. The furtherest couples of robots 
with respect to the barycentre
(if their number is greater than one) probabilistically 
move further from the barycentre along the line defined by themselves 
and the barycentre. Two robots $r_i$ and $r_j$ belong to the set of the $Far\ Robots$ if dist(i, j) $\geq$ dist(w,k) $\forall$ w,k robots in the system. With high probability the above strategy converges to a configuration where the set $Far\ Robots$ contains an unique couple of robots.  
\paragraph{Agreement on the second axis.}
The construction of this second axis is conditioned by the existence of 
two unique nodes. We chose these two nodes as follows: one is the centre of smallest 
enclosing circle while the second one is given by any leader election algorithm.  
Several papers discuss the election of a leader (e.g. \cite{CG07,DP07}). 

\paragraph{Axis orientation and flock head emergence.}
In order to orient the axis we first align the two $Far\ Robots$, $R_A$ and $R_B$, and the leader (see Figure \ref{alignment}). Once the alignment is performed, the first is oriented instructing the leader to create a disimetry between the two points.

\begin{figure}
\begin{center}
\includegraphics[scale=0.5]{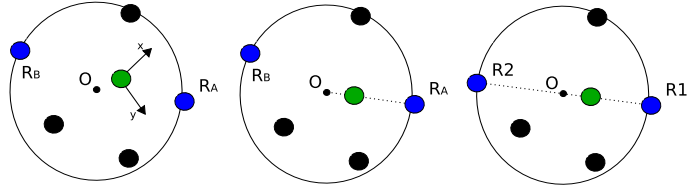}
\end{center}
\caption{Alignment of $Far \ Robots$ and Leader}\label{alignment}
\end{figure}

The alignment strategy is as follows.
If the leader is not aligned with the other two robots then
it will choose between the two robots belonging to the $Far \ Robots$ the one with a bigger value of $x$. In case of symmetry, a bigger value of $y$. Then, it will move toward the intersection of the radius of that robot and the circumference of the cercle with center in $O$ (the center of the $SEC$) and the radius equal to $dist(O, Leader)$.
Finally,  the robot not chosen by the Leader (referred as $R_B$)  will align with the other two robots. $R_B$ moves only when the Leader is aligned with  $R_A$. That is,  when one of the two Far robots, sees that the Leader is aligned with the other $Far \ Robot$ and $O$, then it moves following the $SEC$  until it forms a line with the Leader, the center of the SEC and the other robot in $Far \ Robots$. Note that the two $Far \ Robots$ are on the $SEC$.

For now on,  the $Far  \ Robots$ nearest to the Leader will be referred as $R1$  and the other one $R2$.  Robot $R1$ will play the {\bf flock head} role.

\section{Pattern Emergence}
\label{sec:flocking}
In this section we address the formation of the flocking pattern. Note that we work in a system where robots do not have a common coordinate system. The previous section propose strategies to uniquely identify 3 robots that altogether with the center of the SEC define a common coordinate system. In the following, we assume that over the initial set of $n$ robots 3 robots (referred in the previous section Leader, R1 and R2) are reserved for maintaining the coordinate system while the other $n-3$ can be placed in any shape that will be further reffered as the flocking pattern.
However, we impose a condition on the placement of the shape with respect to the position of the robots that define the references (Leader, R1 and R2) in order to preserve both the unicity of the references and the common coordinate system. 

\begin{lemma}
\label{lemma:1}
The area where the pattern is placed has to satisfy the following three conditions in order to preserve the common coordinate system defined by Leader, R1 and R2. 
\begin{enumerate}
\item All robots must be inside the circumference having $\overline{R1R2}$ as diameter.
\item All the robots must be in the side of the $SEC$ with $R2$ and $y$ negative. 
\item The circle with radius $dist(Leader,O)$ and center in $O$ must be empty.
\end{enumerate}
\end{lemma}

\begin{proof} 
If a robot moves outside the $SEC$, at the next round\footnote{A round is a fragment of execution where each robot in the system executes its actions} the 
$SEC$ will change and consequently the references.  It follows the necessity of condition one. If there exist robots in symmetrical positions with
respect to the $x$ axis, then the system loses the capability to distinguish $R1$ and $R2$. 
This proves the necessity of point two. Point three is motivated by the need of an unique leader. If a robot
goes closer to the center of the $SEC$ than the leader then it becomes $Leader$ and so the 
references will change.
\end{proof}

In order to realize the flocking additional constraints on the shape of the area where the pattern is deployed are needed. These conditions will be discussed in the next section.

The flocking pattern  is obtained in two steps. A first step called bootstraping and a second step that is basically a colission free strategy to move to the pattern positions.
In between these two phases the Leader will move perpendiculary to the segment $R1R2$ in order to orient the second axis. This orientation will be used by the robots in defining a total order among them.

\paragraph{Pattern Bootstraping.}
The bootstraping process takes two phases. In the first phase, all robots but the leader are placed 
on the smallest enclosing circle(SEC). First, the robots closest to the boundaries of the SEC are placed, then recursivelly the other robots.  The algorithm avoids collisions and ensures that robots preserve the referenced (e.g. Leader, R1 and R2) computed in the previous section. In the second phase, the robots on the SEC but R1 will be placed on 
the semi-circle not occupied by the R1 as follows. $R1$ is 
in the position $SEC \cap [O,Leader)$ and $R2$ is on the opposite side of the $SEC$. The other robots are disposed on the quarter of circle around $R2$.
 During this process the references are used  to help the deployement of the others robots. Also, the movement of robots is done such that the semantic of the references is preserved. The code of the algorithms and their analysis are proposed in Annexes, Section \ref{sec:bootstrapping}.

\paragraph{Flocking Pattern Formation.}
 In the following we defined the flocking pattern robots can form in order to maintain the common coordinate system and the references defined in the previous section. 

The flocking pattern $\mathcal{F}$ = $\{ p_1, p_2, \cdots,p_{n-3}, p_{o}, p_{R2} \}$ is the set of points given in input to the robots. It has two distinguished points $p_{o}$ and $p_{R2}$. We call the two distinguished points the $Anchor \ Bolts$ of the pattern, that will correspond to the position of robots $R2$ and to the point ($O$, $dist(Leader,O)$) of the common coordinate system. 
Note that robots start this phase in the following configuration: The segment $OLeader$ is perpendicular to the segment $R1R2$ (where $O$ is the center of the SEC) and all the other robots are disposed on the quarter of circle around $R2$.
In order to form the flocking pattern in a colision free manner robots adopt the following strategy.
First, all robots but the references hook the pattern,  eventually scale  and rotate it   to $R2$ and to the point ($O$, $dist(Leader,O)$). Then they totally order the set of robots based on their coordinates in the common coordinate system and associate to each robot a position in the pattern. Since the order is based on the common coordinate system, all robots will define the exact same order. Then robots move following their order hence collisions are avoided.
In Figure \ref{fig:1} robot $r_2$ has a trajectory that intersects the trajectory of both robots $r_1$ and $r_3$.
However collisions are avoided since following the total order $r_1$ moves first, then  $r_2$ and finally $r_3$.
The detailed code of the algorithm and its analysis are proposed in Annexes, Section \ref{sec:pattern formation}. 
 
\begin{figure}
\begin{center}

\includegraphics[scale=1]{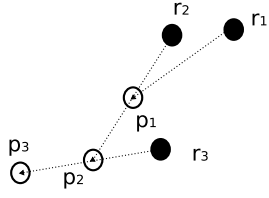}
\end{center}
\caption{Pattern formation strategy}\label{fig:1}
\end{figure}

Once the flocking pattern is formed the Leader moves  to the center of the $SEC$.
This last movement brings the robots in what will be called latter  {\bf Flocking  Formation}.  The motion of this formation will be studied in the next section. 

\section{Velocity Agreement}
\label{sec:motion}
In this section we propose a strategy to synchronize the $Flocking\ Formation$ and the head of the flock. 
The idea is to use $R1$ and $R2$ as two ends of a virtual spring. The other robots will be "pulled" by $R_1$ and "pushed" by $R_2$. Any time $R_1$ or $R_2$ moves, the center of the $SEC$ changes. It follows that  the $Flocking\ Formation$  is not valid since the Leader position in not anymore on the center of the $SEC$. Then, the motion of $R_1$ and $R_2$ is blocked until the flocking pattern is reformed.  The Leader moves 
back to the center of the SEC which makes the $Focking\ Formation$ valid and deblocks 
the motion of robots $R_1$ and $R_2$. 

In the following we precisely characterize two safe regions $\mathcal{M}$ and $\mathcal{K}$.  
$\mathcal{M}$ is the zone where $R1$ is allowed to move and
$\mathcal{K}$ is the area where the pattern can be disposed. In the following we propose 
a general definition of $\mathcal{M}$ and $\mathcal{K}$. 



\begin{definition}
\label{AreaM}
\label{AreaK}
Let $\mathcal{M}$ be the area with $y \geq \pm(kx+R1)$. 
Let 
$\mathcal{K}$ be the area   
between the axis $y=\pm h^{\prime}x$ for $y<0$ and $y=\pm hx + R2$ for $y<0$.
\end{definition}

We additionally define two particular angles, $\alpha$ and  $\beta$ between the axis that border  $\mathcal{M}$ and $\mathcal{K}$. Latter, we prove that $\alpha$ and $\beta$  should 
be greater than $90^0$ in order to verify the conditions stated in Lemma \ref{lemma:1}.

\begin{definition} Let $\alpha$ be the angle between $y= \pm hx + R2$ and  $y =\mp kx+R1$.
Let denote by $A$ the intersection of 
these two lines.  Let $\beta$ be the angle between $y=\pm h^{\prime}x$ and $y= \mp kx$ (see Figure \ref{fig:safe-area}).
\end{definition}

\begin{figure}
\begin{center}
\includegraphics[scale=.4]{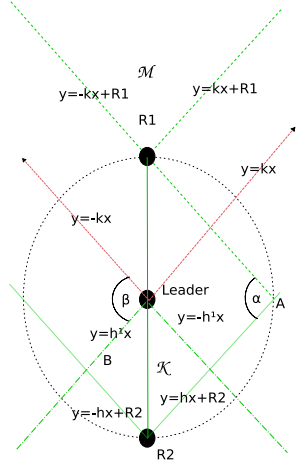}
\end{center}
\caption{Angles $\alpha$ and $\beta$ and the safe areas $\mathcal M$ and $\mathcal K$} 
\label{fig:safe-area}
\end{figure}

The flocking algorithm, Algorithm \ref{alg:PatMot} is executed only when the {\it Flocking Formation} is reached (see Section \ref{sec:flocking}). The algorithm foresee the movement of robots $R1$, $R2$ and the $Leader$.
$R1$ is the robot that imposes the direction of the movement so it can move to any point in the $\mathcal{M}$ area. When it moves, thanks to the constraints we have imposed, the references hold steady and so, at the next observation, all the robots can recognize the references: Leader, R1 and R2. Additionally, all the other robots will be inside the $SEC$ (on the $R2$ side). Then,  all the other robots but $R2$  execute the flocking pattern formation algorithm (presented in Section \ref{sec:flocking}) to align the pattern to the new direction (defined by the axis $R2R1$). Once the $Flocking \ Formation$ is recreated the robots $R1$ or $R2$ can move again. 

When the distance between $R1$ and $R2$  is greater than some parameter $d_{Rmax}$, then $R2$ moves inside the $\mathcal K$ area (along the $R2R1$ segment) within distance $d$. Following Lemma \ref{NoOutSEC} below this distance should be less or equal than $(\frac{dist(R1R2)}{2} - \frac{dist(R1B)}{2cos\delta})$  where $B$ is the closest robot to $R2$ and $\delta$ is the angle  $\angle R2R1B$.

\begin{algorithm} 
\begin{quote}
\begin{tabbing}

1) \textbf{if} \= (Robots form the {\bf Flocking Formation})\\
\> \textbf{If} \= $(dist(R1,R2)<d_{Rmax}$\\
\> \> \textbf{then} R1 moves to a point $\in \ \mathcal{M}$ ;\\
2) \> \textbf{else if} $(dist(R1,R2) \geq d_{Rmax}$\\
\> \> \textbf{then} R2 moves within distance $d$ along the segment $R2R1$;\\
3) \textbf{else} Leader moves perpendicular to $R1R2$ at the center of the SEC;\\
                             then robots execute Flocking Formation algorithm (Section \ref{sec:flocking}).
\end{tabbing}
\end{quote}
\caption{The motion of the Flocking Formation with parameters $d$ and $d_{Rmax}$.}
\label{alg:PatMot}
\end{algorithm}

In the sequel we  determine the relation  between the axis that define the areas $\mathcal{M}$ and $\mathcal{K}$ (i.e. angles $\alpha$ and $\beta$) so that after each movement of $R1$ or $R2$ all the references are preserved.
Firstly, we must guarantee that the $SEC$ will change coherently with the movement of $R1$ and $R2$. At each step the $SEC$ corresponds to the circumference having $R1$ and $R2$ as diameter. 

\begin{lemma}
\label{NoOutSEC}
Let $R1^{\prime}$ be the point where robot $R1$ moves ( inside the $\mathcal{M}$ area). The circle having as diameter $R1^{\prime}R2$ contains all the robots if the angle $\alpha$ is at least $90^0$.
\end{lemma}

\begin{proof}
Consider the worst case: $R1^{\prime}$ belongs to $y= \pm kx +R1$ and there exists a 
robot $B$ $\neq$ $R2$ on the border of the $\mathcal M$ areas: $y=\pm hx+R2$.
Without restraining the generality consider the case where $R1^{\prime}$ moves on the segment 
$y= -kx +R1$ and $B$ is on the line $y=hx+R2$.
When $R1^{\prime}$ diverges from $R1$ the circumference of the new SEC defined by the point $R1^{\prime}$ and $R2$ intersects the line $R1B$ in a point $T$.  
When $\alpha < 90^0$ and $R1^{\prime}$ diverges from $R1$, $T$ moves inside the segment $R1B$ towards $R1$. Hence, at least one robot in the formation (robot $B$) will be outside the new position of the SEC.
When $\alpha \geq 90^0$ and $R1^{\prime}$ diverges from $R1$ the $T$ diverges from $B$ hence every robot in the formation will be inside the new SEC.
\remove{
It's trivial to show that the circumference diverge monotonically from the chord passing for $R2$ and $B$ until the middle point of the chord.
Now, since $\overline{BR2}$ is exactly the half part of that chord having $R1$ as starting position, and when it moves, the chord will tend, we can say that the circumference will ever diverge from $\overline{BR2}$. So, since the circumference and the segment have more than one common point ($R2$), the tangent to the circumference in $R2$ and the axis on $\overline{BR2}$, must be the same. But this is true only if the radius (and then the diameter) passing for $R2$ is perpendicular to $\overline{BR2}$. But if $R1^{\prime}$ go to infinity, this condition will be 
true only if $\alpha = 90^o$.  Now, as we know that $R1$ can make only finite movements even if also extremely big,  the diameter will never be perpendicular to $\overline{BR2}$ and so $\overline{BR2}$ and the circumference will never have other common point then $R2$.}
\end{proof}

\begin{lemma}
\label{NoOutSECR2}
After the movement of  $R2$, the circle having as diameter $R1$$R2^{\prime}$, contains all the robots if $d \leq (\frac{dist(R1R2)}{2} - \frac{dist(R1B)}{2cos\delta})$ where $B$ is the closest robot to $R2$ and $\delta$ is the angle  $\angle R2R1B$. 
\end{lemma}

\begin{proof}
According to Algorithm $\ref{alg:PatMot}$, $R2$ will move on the $y$ axis within distance $d$ from its current position. Let this position be $R2^{\prime}$. Now we will find a value of $d$ such that any robot inside the $\mathcal{K}$ zone, is always inner (or at least on) the circumference having $R1 R2^{\prime}$ as diameter. 
Consider the robot $B$ such that after $R2$ moves,  $B$ is on the border of the new SEC and no other robot is outside the new SEC. Let $\delta$ be the angle between $R2, R1$ and $B$. Using simple geometrical constructions it follows that $d$, the maximal distance $R2$ can move, should less or equal than $(\frac{dist(R1R2)}{2} - \frac{dist(R1B)}{2cos\delta})$.

\end{proof}

Note that after the movement of $R1$ or $R2$  they are still in the set $Far \ Robots$. 
 

\begin{figure}
\begin{center}
\includegraphics[scale=.4]{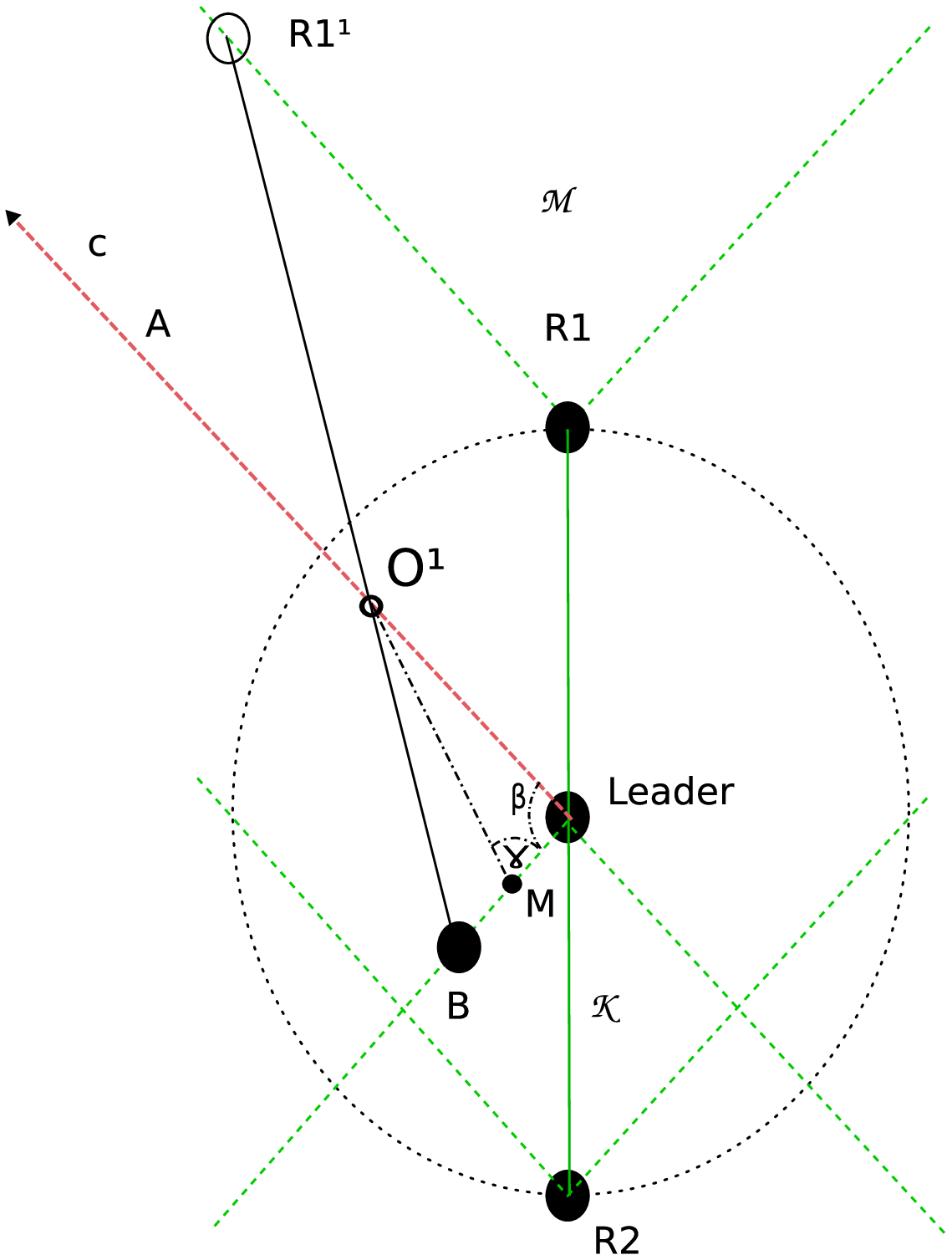}
\end{center}
\caption{Angle $\beta$}\label{AngleBeta}
\end{figure}

In the following we identify a second relation between the axis defining the areas $\mathcal M$ and $\mathcal K$ in order to verify the conditions  of Lemma \ref{lemma:1}.

\begin{lemma}
\label{TheSameRefLead}
After the movement of $R1$ or $R2$, the Leader is preserved and also the references, if $\beta \geq 90^0$ (Figure $\ref{AngleBeta}$).
\end{lemma}

\begin{proof}
Now we will find the smallest value of $\beta$ such that the $Leader$ is preserved.
First, note that the pattern formation algorithm described in the previous sections starts only if the $Leader$ is perpendicular to the $R1R2$ axis, with respect the center of the $SEC$. So, if one of $R1$ or $R2$ moves, the only one that can move after them is the $Leader$. Then, all the other robots must wait until it reaches its final position. After the Leader motion it is still the Leader since it is the closest to the SEC.

Note also that if $R2$ moves, the Leader is always preserved independently of the value of $\beta$. In this case the new center $O^{\prime}$ of the $SEC$ will ever be on the $R1R2$ axis, between the $Leader$ and $R1$. Once again the $Leader$ is preserved since it is the closest robot to $O^{\prime}$.

Cosider the case $R1$ moves. In the worst case $R1$ moves on the border of  $\mathcal{M}$.
Let $R1^{\prime}$ be the next position of $R1$. At each point $R1^{\prime}$ corresponds  a new point $O^{\prime}$ (the middle point of the segment $\overline{R1^{\prime}R2}$) which is the center of the new SEC. If $R1^{\prime}$ moves on the axis $y \geq \mp kx+R1$, then $O^{\prime}$ moves on the parallel axis $y \geq \mp kx$ and if $R1^{\prime}$ goes to infinity, then $O^{\prime}$ also goes to infinity.

Let $B$ be a robot on one of the borders of $\mathcal{K}$, $y=\pm h^{\prime}x$ for $y<0$.
In the following we determin the value of $\beta$ in order to preserve the $Leader$. 
That is, there is no robot closer to $O^{\prime}$ than the $Leader$.
Assume robot $B$ and $Leader$ are equidistant with respect to $O^{\prime}$.
Let $M$ be the middle point of the segment $BLeader$.
Notice that :

1) if  $Leader$ and $B$ are equidistant to $O^{\prime}$ 
then the triangle ($B,O^{\prime},Leader$) is isosceles.

2) if the triangle ($B,O^{\prime},Leader$) is isosceles then the angle $\gamma$ between $Leader,M,O^{\prime}$ is right. It follows that $\beta <90^0$.


So, in order to never have $\overline{O^{\prime}Leader} \leq \overline{O^{\prime}B}$, $\beta \geq 90^0$.

\end{proof}



\begin{lemma}
Starting in any configuration or after any movement of robot $R1$ or robot $R2$ the
the system converges to the flocking specification in O(n) steps.  
\end{lemma}
\begin{proof}
Starting in a configuration that is not a $Flocking \ Formation$ the system converges in $O(n)$ to a $Flocking \ Formation$(the analysis is provided in the Annexes). 
After the movement of either $R1$ or $R2$, the center of the $SEC$ will change. It follows that the  $Leader$ is not anymore on the center of the SEC and the $Flocking \ Formation$ is invalide. Now, due to the first condition in Algorithm \ref{alg:PatMot} any movement of $R1$ or $R2$ is impossible. Then the flocking formation algorithm is executed until a new $Focking \ Formation$, consistent with the new references, will be reached. Following the analysis provided in the Annexes this takes O(n) steps. Then, $R_1$ and $R_2$ are again free to move.
\end{proof}

\section{Conclusions and open questions}
The current work studies the flocking problem in the most generic settings: asynchronous robot networks, oblivious robots, arbitrary initial positions. We proposed a solution with nice self$^\ast$ properties. That is,  in the event of the head leave the flock agrees on another head. Moreover,
robots are oblivious, can start in any initial configuration and do not share any common knowledge.   Additionally,  the flock follows the head whatever its trajectory. Also, the algorithm makes sure that the flock will follow the same head (once emerged) and the system of coordinates does not change during the execution. To this end we identified the necessary conditions that both the pattern and the head velocity have to satisfy in order to maintain the flock pattern, the same unique head and the same coordinate system.

A nice extension of this problem would be to use energy constraints as  in  biological systems. There, in order to conserve the energy of the group,  the head is replaced from time to time. Including energy considerations in the model is a challenge in itself. Also, assuming that heads may change in order to conserve the energy of the group it is also assumed that some kind of common knowledge is shared by all the members of the group. This common knowledge may help in bypassing the impossibility results related to symmetric configurations.
Another interesting extension would be the volumic model. The current solution cannot work in these settings since it is based essentially on alignement properties.

{\bf Acknowledgements}
The last author would like to thank Ted Herman for helpfull discussions related to  flocking in biological systems that inspired the current specification and also the open question stated in the Conclusion section.

\bibliographystyle{plain}
\bibliography{reference}
\newpage
\section{Appendix}
\section{Setting up a common coordinate system}
\label{sec:common-coordinates-details}
The construction of a common coordinate system is built gradually. First, robots 
agree on one axis, then they agree on the second axis. 

\subsection{Agreement on the first axis}
One axis is defined by two distinct points. 
In order to get two common points, robots move in order to distinguish 
an unique couple. The algorithm idea is very simple: robots 
compute the barycentre of their convex hull. The furtherest robots 
with respect to the barycentre
(if their number is greater than two) probabilistically 
moves further from the barycentre along the line defined by themselves 
and the barycentre.

\begin{algorithm}		
	\begin{quote}

\textbf{Functions:\\}
\textit{Far}(myself) returns true if $\exists r_i, d(myself,r_i) \geq d(r_w,r_k) \forall w,k$ robots\\
\textbf{Procedure:\\}
\textit{FarRobots()} returns the set of robots, $r$, Far(r)=true\}
	\begin{tabbing}
		1) Compute the distance $d(myself,r_i)$ between myself and each robot $r_i$, \\
		2) FarRobots()\\
		3) \textbf{if} \= ($myself$ $\in$ \textbf{FarRobots} $\wedge$ $\parallel \textbf{FarRobots} \parallel >2$)\\
			\>  \textbf{then}  $\{$ \=compute the baricentrum of FarRobots()\\
					\> \> move away from the barycentre with probability \\ \> \>																					$p>0$ of a distance $d\_myself \cdot p$\}  
\end{tabbing}
\end{quote}
\caption{Robot Separation executed by robot {\it myself}}
\label{alg:RobSep}
\end{algorithm}

\begin{definition}
Two robots $r_i$ and $r_j$ belong to the set of the \textbf{Far Robots} if dist(i, j) $\geq$ dist(w,k) $\forall$ w,k robots in the system.
\end{definition}

\begin{definition}[legitimate configuration]
\label{legconf}
A legitimate configuration for Algorithm \ref{alg:RobSep} is a configuration with only two \textbf{Far Robots}.
\end{definition}

\begin{lemma}
Starting from a configuration with more than two robots belonging to the set \textbf{Far Robots}, the system executing Algorithm \ref{alg:RobSep} converges to a legitimate 
configuration (see Definition \ref{legconf}) with high probability.
\end{lemma}

The proof follows the same lines as the leader election proof in \cite{CG07}.  

\subsection{Agreement on the second axis}
The construction of this second axis is conditioned by the existence of 
two unique nodes. We chose these two nodes as follows: one is the centre of smallest 
enclosing circle while the second one is given by any leader election algorithm.  
Several papers discuss the election of a leader. In \cite{CG07} the authors prove 
the impossibility of deterministically electing a leader without additional assumptions 
and propose probabilistic solutions. In \cite{DP07} discuss the conditions to deterministically elect a leader. 

\subsection{Axis orientation and head emergence}
The following strategy, referred as Algorithm \ref{alg:alignment} allows the alignment of points computed with Algorithm \ref{alg:RobSep} with respect to the centre of the smallest enclosing circle and the leader. Once the alignment performed this first axis can be easily oriented  using the position of the leader.

\begin{algorithm} 
\label{alg:alignment}
\begin{quote}

\textbf{Functions:\\}
\textit{Lined(A,B,C)} returns true is the three points form a line. \\
\textit{Lined(A,B,C,D)} returns true is the four points form a line. \\
\textit{$choose(Far Robots)$} returns one between the \textit{Far robots}: the far robot with bigger value of $y$ or a bigger value of $x$, in case $y$ values are equal \\

\textit{Note}: We will call the two far robots $R_A$ and $R_B$ only to the sake of precision in the algorithm description. While executing the algorithm robots make no distinction between them.

\begin{tabbing}
1)\textbf{if} \= (R $\in$ \textit{Far Robots} and R  not on the $SEC$)\\
\> \textbf{then}  move to the $SEC$.\\
2) \textbf{if} \= $(\neg Lined(R_{A},R_B,Leader,O)$)\\
3) \>\textbf{if} \= $(\neg Lined(R_{A},Leader,O)$ and $\neg Lined(R_B,Leader,O)$ and $R=Leader$)\\
\> \> \textbf{then} move towards $\overline{OX}$, $X=choose(Far Robots)$ at a distance $dist(Leader,O)$. \\
4) \> \textbf{if} \= $(Lined(R_A,Leader,O)$ and $\neg Lined(R_B,Leader,O)$ and $R=R_B$\\
\> \> \textbf{then}  move towards $\overline{R_ALeader} \cap SEC$. \\

\end{tabbing}
In the sequel we will call $R1$ the $Far  \ Robot$ closest to $Leader$ and $R2$ the other one.
\end{quote}
\caption{Alignment algorithm executed by robot $R$}
\end{algorithm}

The first operation is to make sure that the two elected $Far \ Robots$ are on the $SEC$. Otherwise, they move to the SEC (smallest enclosing circle).
 
If the leader is not aligned with the other two robots then
it will choose between the two robots belonging to the $Far \ Robots$ the one with a bigger value of $x$. In case of symmetry, a bigger value of $y$. Then, it will move toward the intersection of the radius of that robot and the circumference with center in $O$ (the center of the $SEC$) and radius equals to $dist(O, Leader)$.
Finally,  the robot not chosen by the Leader (referred in the algorithm is $R_B$)  will align with the other two robots. $R_B$ moves only when the Leader is aligned with the other robot. That is,  when one of the two Far robots, sees that the Leader is lined up with the other $Far \ robot$ and $O$, then it moves following the $SEC$  until it form a line with the Leader, the center of the SEC and the other robot in $Far \ Robots$. Recall  that the two $Far \ robots$ moved previously to the $SEC$.

For now on,  the $Far  \ Robots$ nearest to the Leader will be referred as $R1$ (Reference 1) and the other one $R2$. 

\begin{lemma}
During its movement, $R2$, will never collide with another robot.
\end{lemma}

\begin{proof}
Assume $R2$ collides with a robot $R3$.
$R3$ will belong to the $SEC$ and it will be closer than $R2$ to the point at the end of the diameter of the $R1$, so the the angle $\alpha_3$ between diameter of $R1$ and $\overline{R1R3}$ will be smaller than the angle $\alpha_2$ between diameter of $R1$ and $\overline{R1R2}$. But distance(R1, R2) is equal to $diameter \cdot cos(\alpha_2)$ that is smaller than  $diameter \cdot cos(\alpha_3)$. It follows that $R3$ will be farther than $R2$ with respect to the $R1$ position which contradicts 
the hypothesis ($R_2$ is the furthest robots with respect to $R1$).
\end{proof}

\section{Bootstraping the flocking pattern}
\label{sec:bootstrapping}
In this section we execute a pre-processing that prepares the set up of
 the flocking pattern used further by the flocking 
algorithm. We build on top of the algorithms described in the previous section.
The pattern includes three robots used as references (called Leader, R1 and R2) and other n-3 robots that can be placed in any shape.  

The bootstraping process takes two phases. First, all robots but the leader will be placed 
on the smallest enclosing circle. Then, the robots on the SEC but R1 will be placed on 
the semi-circle not occupied by the leader.

\subsection{Phase 1: Placement on the Smallest Enclosing Cercle}
 The idea of the algorithm proosed as Algorithm \ref{alg:positioning} is very simple. First, the robots closest to the boundaries of the smallest enclosing circle are placed. Than recursivelly the other robots.  The algorithm avoids collisions and ensures that robots preserve the referenced (e.g. Leader, R1 and R2) computed in the previous section.

\begin{definition}
Let $FreeToMove$ be the set of robots without robots between themselves and the 
$SEC$ (including the border) along the radius passing through them, and that does not belong to the $SEC$.
\end{definition}

\begin{definition}
Let $AlreadyPlaced$ be the set of the robots belonging to the border of the $SEC$. 
\end{definition}

\begin{algorithm}[h]
\begin{quote}

{\bf Preprocessing:}\\
$\forall r_i$ compute the value of the radius passing through $r_i$. Let $rad_{r_i}$ be the 
value of the angle between my radius ($rad_{myself}=0$) and the radius of robot $r_i$, in clockwise direction (note that clockwise depends only on the coordinate system of each robot) \\
$\forall r_i$ compute the value of $dist_{r_i}$, distance of the robot $r_i$ to the border of the smallest enclosing circle ($SEC$)\\

{\bf Predicates:} \\
$Leader(myself) \equiv \forall r_i$ with $i \neq myself$, $dist_i < dist_{myself}$  \\

{\bf Functions:}\\
$OccupiedPosition(rad_{myself}):$ returns $r_i, i \neq myself, dist_{r_i}=0 ~and ~ rad_{r_i}=rad_{myself}$ otherwise $\bot$ \\
$NextToMove:$ returns the set of closest robots $r$ to the $SEC$ with $dist_{r} \neq 0$\\

1) \textbf{if} $\neg Leader(myself) \wedge myself \in FreeToMove$\\  
	\textbf{then} $\{$ move to $SEC$ with distance $dist_{myself}\}$\\
2) \textbf{if} ($\neg Leader(myself) \wedge (myself \in NextToMove) \wedge (FreeToMove=\emptyset) 
\wedge (OccupiedPosition(rad_{myself}) \neq \bot$))\\
	\textbf{then} $\{$ Move to the first quarter point of the arch between robot OccupiedPosition($rad_{myself}$) 
and robot $r_j$ belonging to the $SEC$ such that $rad_j$ is minimum.$\}$
\end{quote}  
\caption{Placement executed by robot $my\_self$}
\label{alg:positioning}
\end{algorithm}

\begin{definition}
A legitimate configuration for Algorithm \ref{alg:positioning} is a configuration 
where all robots but the leader are in the set $AlreadyPlaced$. 
\end{definition}

Note that the algorithm does not change the leader position neither the position of $AlreadyPlaced$ nodes ($R1$ and $R2$ belong to the set of the $Already Placed$).
Moreover, there is no node behind the leader and sharing the same radius as the leader. Otherwise 
this node will be the closest to the center of the $SEC$.

\begin{lemma}
Algorithm \ref{alg:positioning} is silent.
\end{lemma}

\begin{proof}
The leader does not change its position and
once all robots but the leader are in the set
$AlreadyPlaced$ no robot can execute its actions so the algorithm is silent.
\end{proof}

\begin{lemma}
\label{lemma:freetomove}
If two robots $r_i$ and $r_j$ belong to the set $FreeToMove$, than their final position will be different.
\end{lemma}

\begin{proof}
If there are no robots between $r_i$ and the $SEC$ along $radius_i$ and there are no robots between 
$r_j$ and the  $SEC$ along $radius_j$ than $radius_j$ and $radius_i$ must be different. 
To different radius correspond different positions on the $SEC$ so $r_i$ and $r_j$, thanks  to Rule 1, 
will be placed on different positions.
\end{proof}

\begin{lemma}
\label{lemma:freepositioning}
A robot always moves towards a free position on the $SEC$.\\
\end{lemma}

\begin{proof}
Rule 1 moves robots in the set $FreeToMove$. By definition  
these robots move only on free positions on the $SEC$.

Some robots not in the set $FreeToMove$ are allowed to move only when the set $FreeToMove$ 
is empty and they are in the set $NextToMove$.
Let $r_i$ and $r_j$ be such robots. $r_i$ and $r_j$ are enabled for the Rule 2 and move towards 
the quarter point of the arc between the robot $r_{ii}$ (and $r_{jj}$) and their next robot on the border of $SEC$. The worst case is if $r_i$ and $r_j$ can move together and their own system of coordinate is such that the next robot for $r_i$ is $r_{jj}$ and the next for $r_{jj}$ is $r_{ii}$ (see Figure \ref{fig:positioning}).

Firstly we can notice that the arc between the two robots is free before their movement ( otherwise  $r_j$ would not be the next one of $r_{ii}$ and vice versa). 
Moreover, the sector between $r_{jj}$ and $r_{ii}$ is free.
Assume a robot exists in that area. Following Rule 2 it must  go towards $SEC$ before $r_i$ moves and hence becoming the next robot of the $SEC$ after $r_{ii}$. 
Now $r_i$ and $r_j$ can move freely in this sector, the only problem can be due to a reciprocal collision.
But, since $r_i$ and $r_j$ can move only in the first quarter of their area, starting from an opposite side (the $\mathcal{K}$ zones of Figure 2), then they cannot collide. 
If instead one of the robots cannot reach the quarter point in one move (due to scheduler interference), at its next schedule it will still verify the conditions of Rule 1, so it will move to the $SEC$ following this same rule. 	
\end{proof}

\begin{figure}
\begin{center}
\begin{pspicture}(10,5)
\pswedge[fillstyle=vlines, linecolor=white](5.24,5.24){4.8}{283}{318}
\pspolygon*[fillcolor=white, linecolor=white](5.24,5.24)(6.2,0.5)(7.05,3.6)

\pswedge[fillstyle=vlines, linecolor=white](5.24,5.24){4.8}{211}{245}
\pspolygon*[fillcolor=white, linecolor=white](5.24,5.24)(3.35,4.1)(3.2,.9)

\psarc{<-<}(5.24,5.24){4.8}{200}{0}
\uput{.1}[225](5.3,5.6){$O$}
\pscircle*(8.8,2){0.25}
\uput{.1}[225](9.8,2){$r_{ii}$}
\psline (5.24,5.24)(8.8,2)
\pscircle*(1.1,2.8){0.25}
\uput{.1}[225](1,2.3){$r_{jj}$}
\psline[linecolor=black] (5.24,5.24)(1.1,2.8)

\pscircle(3.35,4.1){0.25}
\uput{.1}[225](3,4.5){$r_j$}
\pscircle*(3.2,.9){0.25}
\uput{.1}[225](3,.5){$r_j$}
\psline[linecolor=red]{->}(3.35,4.1)(3.2,.9)

\pscircle(7.05,3.6){0.25}
\uput{.1}[225](7.65,4){$r_i$}
\pscircle[linestyle=dotted](6.2,0.5){0.25}
\pscircle(6.56,2){0.25}
\uput{.1}[225](6,2.1){$r_i$}
\psline[linestyle=dotted, linecolor=black] (7.05,3.6)(6.2,0.5)
\psline[linestyle=dashed, linecolor=black] (5.24,5.24)(7,0.85)
\pscircle*(7,0.85){0.25}
\uput{.1}[225](7.5,0.5){$r_i$}
\psline [linecolor=red]{*->}(7.05,3.6)(6.56,2)
\psline [linecolor=red]{->}(6.56,2)(7,0.85)
\psframe [fillstyle=vlines](10,0)(11,.5)
\rput (11.7,.2){\textit{zone} $\mathcal{K}$}
\end{pspicture}
\end{center}
\caption{  }
\label{fig:positioning}
\end{figure}

\begin{lemma}
\label{lemma:collision_free}
Algorithm \ref{alg:positioning} is collisions free 
(two robots never move towards the same free position).
\end{lemma}

\begin{proof}
Firstly if two robots $r_i$ and $r_j$ start at different distances from the $SEC$, only the one closest to SEC can move.
The other one must wait until the former reaches the circumference.
If the two robots $r_i$ and $r_j$ have the same distance with respect to the $SEC$, than they may move at the same time.
If they are enabled for the Rule 1 then they never collide since they will reach different positions on the 
$SEC$ moving along their respective radius which are different (Lemma \ref{lemma:freetomove}).
If both $r_i$ and $r_j$ are enabled for Rule 2
then both  robots must move towards the quarter point between the robot $r_p$ on $SEC$ belonging 
to its radius and the next robot $r_{Pnext}$ on the $SEC$ in their own clockwise direction.

Following Lemma \ref{lemma:freepositioning} the sector between $r_p$ and $r_{Pnext}$ is free from other robots and 
a robot can move only inside the zone $\mathcal{K}$ (see Figure \ref{fig:positioning}). 
Since $r_i$ and $r_j$ do not belong to the same radius so $r_{p_i} \neq r_{p_j}$ and $r_{{Pnext}_i }\neq r_{{Pnext}_j}$.
So the two $\mathcal{K}$ zone where $r_i$ and $r_j$ can move have no intersection.
It follows,  $r_i$ never reaches $r_{{Pnext}_i}$.

Overall $r_i$ and $r_j$ will never reach the same final position nor 
meet on the way to their respective final positions.
\end{proof}

\begin{lemma}
Algorithm \ref{alg:positioning} converges in $O(n)$ steps to a legitimate configuration. 
\end{lemma}

\begin{proof}
Rule 1 allows only robots in $FreeToMove$ to move to the $SEC$.  
Thanks to Rule 2 all robots that don't belong to this set  must wait until that it is empty.
Once this set is empty, robots, excepted the leader, that are not on the $SEC$ can execute the Rule 2. 
This rule is such that it can be executed only by the set of robots $i$ with $min(dist_i)$ and $dist_i\neq0$. 
Now these robots can move towards the middle point of the arc of $SEC$ between the robot on 
$SEC$ belonging to its radius and the next robot on the $SEC$  in clockwise direction.
Once these robots have arrived on the $SEC$ and their corresponding $dist=0$, other robots 
can satisfy Rule 2 and move to the $SEC$.
This process will be iterate until all the robots except the Leader are on the $SEC$.
Following Lemma \ref{lemma:collision_free} robots do not collide and no 
robot obstructs the trajectory of other robots.

In the worst case the algorithm converges in O(n) steps. 
\end{proof}

\subsection{Phase 2: Setting the circular flocking configuration}
This phase is  a pre-processing phase for forming the final motion pattern.
Starting from the final configuration of Algorithm \ref{alg:positioning} the curent phase reaches 
a circular flocking configuration as shown in Figure \ref{fig:figure2}. 
The circular shape has the following caracteristics: $Leader$ is inside 
the $SEC$ (the one computed by Algorithm \ref{alg:positioning}) 
and all the other robots are disposed on the SEC border. 
These robots are placed as follows: $R1$ is 
in the position $SEC \cap [O,Leader)$ and $R2$ is on the opposite side of the $SEC$. The other robots are disposed on the quarter of circle around $R2$.
In the following this configuration will be referred as {\it circular flocking configuration} (see Figure \ref{fig:figure2} for a nine robots example). 

\begin{figure}[h]
\begin{center}
\begin{pspicture}(10,8)

\psarc{<-}(6,5){4}{0}{360}
\uput{.1}[225](6,5.5){$O$}
\pscircle*(8.82,7.82){0.25}
\uput{.1}[225](9.8,7.4){$R1$}
\pscircle*(6.5,5.5){0.25}
\uput{.1}[225](7.2,5.5){$Leader$}
\pscircle*(2.35,3.35){0.25}
\uput{.1}[225](2,3){$p_2$}
\pscircle*(4.43,1.36){0.25}
\uput{.1}[225](4,1){$p_3$}
\pscircle*(6,1){0.25}
\uput{.1}[225](6.5,.6){$p_4$}
\pscircle*(2,5){0.25}
\uput{.1}[225](2.6,4.8){$p_1$}

\pscircle*(5.2,1.1){0.25}
\uput{.1}[225](6.5,.6){$p_4$}
\pscircle*(3.75,1.67){0.25}
\uput{.1}[225](2.6,4.8){$p_1$}

\pscircle*(3.17,2.17){0.25}
\uput{.1}[225](3.9,2.25){$R2$}
\psline[linestyle=dashed, linecolor=black] (8.82,7.82)(3.17,2.17)
\psline[linestyle=dashed, linecolor=black] (2,5)(6,5)
\psline[linestyle=dotted, linecolor=black] (3.17,7.82)(8.82,2.17)
\psline[linestyle=dashed, linecolor=black] {-*}(6,1)(6,5)
\psline[linestyle=dashed, linecolor=black] (6,5)(2.35,3.35)
\psline[linestyle=dashed, linecolor=black] (6,5)(4.43,1.36)
\psline[linestyle=dashed, linecolor=black] (6,5)(5.2,1.1)
\psline[linestyle=dashed, linecolor=black] (6,5)(3.75,1.67)
\end{pspicture}
\end{center}
\caption{Circular flocking configuration}
\label{fig:figure2}
\end{figure}
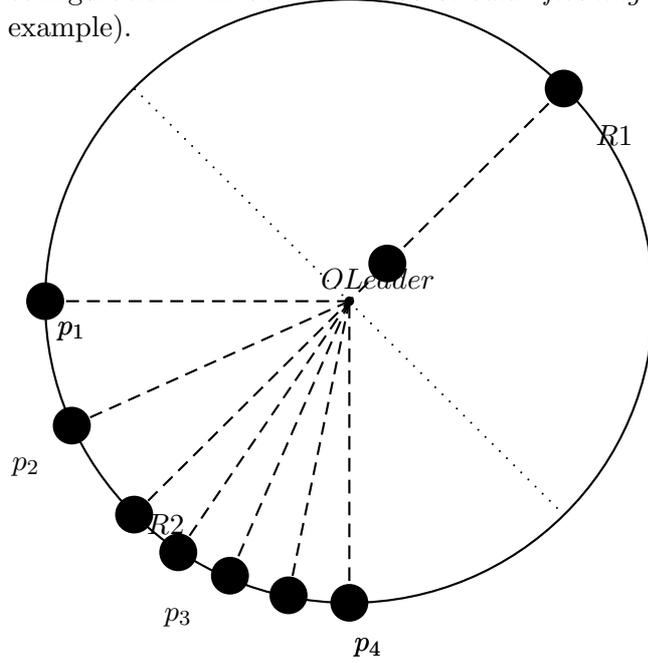

In order to construct the circular flocking configuration
we introduce the concept of oriented configuration:
\begin{definition}[oriented circular configuration]
A configuration is called circular oriented if the following conditions hold:
\begin{enumerate}
   \item All robots are at distinct positions on the  $SEC$ , 
         except only one of them, called  $Leader$, located inside $SEC$ ;
   \item $Leader$ is not located at the center of $SEC$;
	\item Two robots, named $R1$ and $R2$, form with $Leader$ and $O$ (the center of the $SEC$) a line;
\end{enumerate}
\end{definition}

Note that the output of Algorithm \ref{alg:positioning} satisfies point 1 of the definition above.
Note also that the alignment of the references phase supply the point 2 and 3 of the definition above.

Algorithm \ref{alg:flockconfiguration} below started in an oriented configuration  
eventually converges to a circular flocking configuration.

The algorithm makes use of the following function: 
$FinalPositions(SEC, R1, R2)$ which 
returns, when invoked by a robot, the set of positions in the circular flocking configuration
with respect to $SEC$ and to the points $R1$ and $R2$. $R1$ is the 
robot on the intersection between the segment $[O, Leader)$ and the circle $SEC$.
$R2$ is the robot on the intersection between the diameter of $SEC$ passing through $R1$ and the center of $SEC$. To calculate the position of the other robots, we split the circumference in two parts, taking $R1$ and $R2$ as terminal points. Each robot counts how many robots are in its semi-circumference, so it will create, in the first quarter closest to $R2$, as many equidistant final positions as many robots are(considering that the last position is already occupied by $R2$). 
The order of these positions is given from the closest to $R1$, to the furthest 

\begin{algorithm}[h]                                                  
\begin{quote}
Functions:\\
\hspace*{0.5cm}
$get\_number(myself)$ returns the number of robots  
between $myself$ and \\
\hspace*{2.5cm}position $p_1$ (including robot myself) clockwise\\
\hspace*{0.5cm}
$get\_position(myself)$ returns the position $get\_number(myself)$ \\
\hspace*{2.5cm}
in $FinalPositions(SEC,p_1)$\\
\hspace*{0.5cm}
$FreeToMove(myself)$ returns true if there are no robots between $myself$\\
\hspace*{2.5cm} 
and $get\_position(myself)$\\

Motion Rule:\\
\hspace*{1cm}
  {\bf if}  FreeToMove(myself) \textbf{then}\\
\hspace*{1.5cm}
         move to $get\_position(myself)$\\          
\end{quote}
\caption{Setting the motion formation executed by robot $myself$}
\label{alg:flockconfiguration}
\end{algorithm}

The idea of the algorithm is as follows. Robots started in an oriented configuration reach their 
final positions in the circular flocking configuration. 
If a robot is blocked by some other robots than it waits until all these 
robots are placed in their final positions.

\begin{lemma}
In a system with $n$ robots,
Algorithm \ref{alg:flockconfiguration} started in an oriented configuration 
converges in $O(n)$ steps  to a configuration where all robots 
reached their final positions computed via FinalPositions function.
\end{lemma}

\begin{proof}
For this proof we can consider only one semi-circumference, the behavior of the other one will be totally symmetric to the former.
Let $p_1, \ldots, p_n$ be the robots' final positions returned by $FinalPositions(SEC,R1,R2)$  and let $r_1, \ldots, r_n$ be the set of robots to be placed. 
Assume the worst initial configuration: no robot is placed in its final position
Let denote by $L$ the segment defined by the robots which are not placed in their final positions.
The initial length of $L$ is $n$.  
The robots $r_1$ (the first robot in $L$) can freely move to its final position $p_1$.
Once this robot is placed the length of segment $L$ becomes $n-1$. One of the new ends of $L$ 
can be placed to its final position. Assume the contrary. All robots are blocked.
So, there is a waiting chain such that $r_2 \rightarrow r_3 \rightarrow \ldots \rightarrow r_n$ or 
$r_n \rightarrow \ldots r_2$. Since the chain is finite and not cyclic (due to the total order) 
one of the ends of the chain can move ($r_2$ or $r_n$). After this robot moves the length of the segment 
$L$ decreases. Eventually, all robots in $L$ finish in their final positions.
In the worst case, the algorithm converges in $O(n)$ steps. 
\end{proof}

\section{Flocking Pattern}
\label{sec:pattern formation}
In the following we describe the pattern the robots can form. 
\begin{definition}[Pattern]
A pattern $\mathtt{P}$ = $\{ p_1, p_2, \cdots,p_{n-3}, p_{o}, p_{R2} \}$ is the set of point given in input to the robots. It has two distinguished points $p_{o}$ and $p_{R2}$. We call the two distinguished points the $Anchor \ Bolts$ of the pattern, that will correspond to the position of robots $R2$ and to the point ($O$, $dist(Leader,O)$) of the common coordinate system. 
\end{definition} 

First, all robots but the $references$ hook the pattern,  eventually scale  and rotate it, to $R2$ and to the point ($O$, $dist(Leader,O)$). Then they associate to each robot a position in the pattern. Algorithm \ref{alg:PatFor} ensures that each robot goes to its position without colliding with other robots.
Several definitions are needed.
 \begin{definition}
 A robot $r_i$ belongs to the set $Free \ Robot$ if $r_i \neq p_i$.
 \end{definition}
 
\begin{definition}
 A position $p_i$ belongs to the set $Free \ Position$ if $r_i \neq p_i$.
\end{definition}

\begin{algorithm}	
	\begin{quote}
	\begin{tabbing}
	
Functions:\\
	$Trajectory_i$: the segment that joins $r_i$ to $p_i$ \\

	$Next$($P_i$,$P_j$): \= returns true if $x_{P_i} < x_{P_j}$ or, if  ($x_{P_i} = x_{P_j}$) and $y_{P_i} < y_{P_j}$; otherwise returns false\\
We can also say $P_i$ has $P_j$ as Next robot(position)\\\\

Assignments to all the $Free \ Robots$  but the references and to all $Free Position \neq p_L$\\ but the $anchor bolts$ of a sequential number\\ from the robot (position) with the smaller value of $x$ to the bigger.\\ If two or  more robots have the same $x$ value, then consider their $y$ value.\\\\

	1)	\textbf{for} \= (all $r_i$ 	$\neq$ $r_{me})$\\
			\> \textbf{If} \= ($p_{me} \in \ Trajectory_i)$\\
			\>	\> \textbf{then} Exit;\\

	2)	\= \textbf{if} \= (Next($p_{me},r_{me})$)\\
			\>\> \textbf{If} \= (Next($r_{me-1}, p_{me}$))\\
			\>\>	\> \textbf{then} (move along the $Trajectory_{me}$ until $x_{me-1}+\varepsilon$ );\\
			\>\>  \textbf{Else} (move to $p_{me})$\\
		\> \textbf{If} \= (Next($r_{me},p_{me})$)\\
		\>	\> \textbf{If} \= (Next($r_{me+1}, p_{me}$))\\
		\>	\>	\> \textbf{then} (move along the $Trajectory_{me}$ until $x_{me+1}-\varepsilon$ );\\ 
		\>	\>  \textbf{Else} (move to $p_{me})$\\

\end{tabbing}
\end{quote}
\caption{Pattern Formation executed by robot "me"}
\label{alg:PatFor}
\end{algorithm}

\begin{definition}[DeadLock]
 A Deadlock configuration  is a configuration where  each robot is blocked by another robot.
\end{definition}
Note that: 	
1) Following Rule 1 a Deadlock arrives if there exists a robot $i$ on each $trajectory_j$. $ \forall i \neq j$ and
2) Following Rule 2 a Deadlock can occur if all the robots have a Next robot.

In the following we prove that none of these two conditions are satisfied.


\begin{lemma}
\label{NoDeadLock}
No deadlock configuration is reached. 
\end{lemma}

\begin{proof}
In the following we prove that none of the two deadlock conditions are satisfied. 




Assume that the three robots are in a DeadLock situation. We call these three robots $r_1$, $r_2$ and $r_3$. To these three robots will correspond three position $p_1$, $p_2$, $p_3$. Assume also that the following statements are true:
Next($p_1$, $p_2$) and 
Next($p_2$, $p_3$). 

If DeadLock then: 
$p_1 \in Trajectory_2$ and
$p_2 \in Trajectory_3$ and 
$p_3 \in Trajectory_1$.

In the following we prove that none of the previous situations can happen. 
First assume Next($r_2$,$p_2$) returns true\footnote{The case Next($p_2$,$r_2$) is totally symmetric}:
If $p_1 \in Trajectory_2$   then Next($p_1$,$p_2$) returns true .

Let $p_2$ on the $Trajectory_3$  then Next($p_2,p_3$) returns true. 
If $p_3 \in \ Trajectory_1$ then Next($p_3,p_1$) returns true which contradicts the hypothesis Next($p_1$, $p_3$).

From the definition of $Next$ we can assert that the relation is unequivocal (unambiguous) so if we have two robots A and B and Next(A,B) returns true then Next(B,A) must return  false.  
Also if   Next(A,B) returns true  and  Next(B,C) returns true  then also  Next(A,C) should return true. As before we can say that A and B have a next robot but C do not has any Next robot.  If we iterate the process we will always find one robot without a Next robot.
\end{proof}

\begin{lemma}
\label{CollisionFree}
The algorithm is collision free\footnote{Note that is the algorithm is collision free two robots cannot try to get the same position.}. 
\end{lemma}

	\begin{proof}
Easily by the point 2  no robot can surpass another robot, moreover thanks to the same point, two FreeRobots, starting from position with different $x$ values will never occupy  two positions with the same $x$  value. Similarly if two robots start from positions with the same $x$ value, after the activation of one of the two robots, they will be in positions with a different $x$ value. Having at each round position with a different $x$ value, the robots never collide.
\end{proof}

\begin{lemma}
All robots will eventually reach their final position in the pattern.
\end{lemma}

	\begin{proof}

Lemma \ref{CollisionFree} and Lemma \ref{NoDeadLock} guarantee that all robots can move to their own position.
\end{proof}

The last operation, once the pattern is bootstrapped is to bring the Leader to the center of the $SEC$.
This last movement brings the robots in what will be called latter $Flocking \ Formation$.  The motion of this formation will be studied in the next section.

\end{document}